\documentclass{article}[11pt]
\usepackage{amsmath}
\usepackage{amsfonts}
\usepackage{algorithm2e}
\usepackage[noend]{algorithmic}
\usepackage{url}
\newcommand{\qed}{\hfill \ensuremath{\Box}}
\usepackage{chngpage}

\SetKw{Kw}{function}

\usepackage{fullpage}
\usepackage{authblk}

\usepackage{makeidx}

\newtheorem{lemmx}{Lemma}
\newtheorem{thmx}{Theorem}

\newenvironment{proof}{{\bf Proof:}}{\qed}

\newcommand{\E}{ \mathbb{E}}
\renewcommand{\Pr}{\mbox{\rm \bf{Pr}}}

\begin{document}
\title{Triangle counting in dynamic graph streams}
\date{}
\author[1]{Laurent Bulteau \thanks{\tt{l.bulteau@gmail.com}, Supported by the Alexander von
    Humboldt Foundation, Bonn, Germany.}}
\author[1]{Vincent Froese\thanks{\tt{vincent.froese@tu-berlin.de}, Supported by the DFG project DAMM (NI
    369/13).}}
\author[2]{Konstantin Kutzkov\thanks{\tt{kutzkov@gmail.com}, Work done while the author was at IT University of Copenhagen and supported by the Danish National Research Foundation under the Sapere Aude program.}}
\author[3]{Rasmus Pagh\thanks{\tt{pagh@itu.dk}, Supported by the Danish National Research Foundation under the Sapere Aude program.}}
\affil[1]{Technische Universit\"at Berlin, Germany}
\affil[2]{NEC Laboratories Europe, Heidelberg, Germany}
\affil[3]{IT University of Copenhagen, Denmark}
\maketitle

\begin{abstract}
Estimating the number of triangles in graph streams using a limited amount of memory has become a popular topic in the last decade. Different variations of the problem have been studied, depending on whether the graph edges are provided in an arbitrary order or as incidence lists. However, with a few exceptions, the algorithms have considered {\em insert-only} streams.
We present a new algorithm estimating the number of triangles in {\em dynamic} graph streams where edges can be both inserted and deleted. We show that our algorithm achieves better time and space complexity than previous solutions for various graph classes, for example sparse graphs with a relatively small number of triangles. Also, for graphs with constant transitivity coefficient, a common situation in real graphs, this is the first algorithm achieving constant processing time per edge.  The result is achieved by a novel approach combining sampling of vertex triples and sparsification of the input graph. In the course of the analysis of the algorithm we present a lower bound on the number of pairwise independent 2-paths in general graphs which might be of independent interest. At the end of the paper we discuss lower bounds on the space complexity of triangle counting algorithms that make no assumptions on the structure of the graph.
\end{abstract}
\section{Introduction}

Many relationships between real life objects can be abstractly represented as graphs. The discovery of certain structural properties in a graph, which abstractly describes a given real-life problem, can often provide important insights into the nature of the original problem. The number of triangles, and the closely related clustering and transitivity coefficients, have proved to be an important measure used in applications ranging from social network analysis and spam detection to motif detection in protein interaction networks. We refer to~\cite{sparsifiers} for a detailed discussion on the applications of triangle counting. 

The best known algorithm for triangle counting in the RAM model runs in time $O(m^{\frac{2 \omega}{\omega+1}})$~\cite{exact_counting} where $\omega$ is the matrix multiplication exponent, the best known bound is $\omega = 2.3727$~\cite{virgi_mm}. However, this algorithm is mainly of theoretical importance since exact fast matrix multiplication algorithms do not admit an efficient implementation for input matrices of reasonable size.

The last decade has witnessed a rapid growth of available data. This has led to a shift in attitudes in algorithmic research and solutions storing the whole input in main memory are not any more considered a feasible choice for many real-life problems. Classical algorithms have been adjusted in order to cope with the new requirements and many new techniques have been developed. This has led to the {\em streaming} model of computation where only a single scan of the data is possible. Both efficient algorithms and impossibility results have shed light on the computational complexity of many problems in the model~\cite{muthu_survey}. %A relaxation of the model is the so called {\em semi-streaming} model when one is allowed a small number of reads of the input. For many real-life problems this is a feasible setting because one can afford to store the entire input on an external device such that it can be accessed sequentially and each scan is considered expensive. %However, for real streaming applications the semi-streaming applicat
%
%\subsection{Previous work}
%\paragraph{Exact triangle counting.}
%
\\\\
{\bf Approximate triangle counting in streamed graphs.} 
For many applications one is satisfied with a good approximation of the number of triangles instead of their exact number, thus researchers have designed randomized approximation algorithms returning with high probability a precise estimate using only small amount of main memory. Two models of streamed graphs have been considered in the literature. In the {\em incidence list stream} model the edges incident to each vertex arrive consecutively and in the {\em adjacency stream} model edges arrive in arbitrary order. Also, a distinction has been made between algorithms using only a single pass over the input, and algorithms assuming that the input graph can be persistently stored on a secondary device and multiple passes are allowed. %The problem has been extensively studied, and the best known results are summarized in Table XXX. 
A simple approach for estimating the number of triangles in insert-only streams is to sample a certain number of 2-paths, then compute the ratio of 2-paths in the sample that are completed to triangles and multiply the obtained number with the total number of 2-paths in the graphs. For incidence list streams this is easy since we can assume that the stream consists of (an implicit representation of) all 2-paths~\cite{buriol_et_al_1}. For  the more difficult model of adjacency streams where edges arrive in arbitrary order the approach was adjusted such that we sample a random 2-path~\cite{kdd_triangles,pavan_et_al}. %In a recent work Jha et al.~\cite{kdd_triangles} showed how to adjust the approach to the more difficult model of adjacency streams where edges arrive in arbitrary order. 
The one-pass algorithm with the best known space complexity and constant processing time per edge for adjacency streams is due to Pavan et al.~\cite{pavan_et_al}, and when several passes are allowed -- by Kolountzakis et al.~\cite{semistream_counting}. For a more detailed overview of results and developed techniques we refer to~\cite{sparsifiers}.
\\
Dynamic graph streams have a wider range of applications. Consider for example a social network like Facebook where one is allowed to befriend and ``unfriend'' other members, or join and leave groups of interest. Estimating the number of triangles in a network is a main building block in algorithms for the detection of emerging communities~\cite{comm_det}, and thus it is required that triangle counting algorithms can also handle edge deletions. The problem of designing triangle counting algorithms for dynamic streams matching the space and time complexity of algorithms for insert-only streams has been presented as an open question in the 2006 IITK Workshop on Algorithms for Data Streams~\cite{open_question}. The best known algorithms for insert-only streams work by sampling a non-empty subgraph on three vertices from the stream (e.g. an edge $(u, v)$ and a vertex $w$). Then one checks whether the arriving edges will complete the sampled subgraph to a triangle (we look for $(u,w)$ and $(v,w)$). The approach does not work for dynamic streams because an edge in the sampled subgraph might be deleted later. Proposed solutions~\cite{ahn_et_al,cycles_counting,iran_triangles} have explored different ideas. These approaches, however, only partially resolve the open problem from~\cite{open_question} because of high processing time per edge update, see Section~\ref{sec:results} for more details.

\paragraph{Our contribution.}
In this work we propose a method to adjust sampling to work in dynamic streams and show that for graphs with constant transitivity coefficient, a ubiquitous assumption for real-life graphs, we can achieve constant processing time per edge.  
We also show that \emph{some} assumption on the graph is needed to be able to estimate small triangle counts, by showing a lower bound on the space usage in terms of the number of edges and triangles, matching the upper bound of Manjunath et al.~\cite{cycles_counting} for constant approximation factor.

At a very high level, the main technical contribution of the present work can be summarized as follows:
%\\  
%A simple approach for estimating the number of triangles in insert-only streams is to sample a certain number of 2-paths, then compute the ratio of 2-paths in the sample that are completed to triangles and multiply the obtained number with the total number of 2-paths in the graphs. For incidence list streams this is easy since we can assume that the stream consists of (an implicit representation of) all 2-paths~\cite{buriol_et_al_1}. %In a recent work Jha et al.~\cite{kdd_triangles} showed how to adjust the approach to the more difficult model of adjacency streams where edges arrive in arbitrary order. 
For dynamic graph streams sampling-based approaches fail because we don't know how many of the sampled subgraphs will survive after edges have been deleted. On the other hand, graph sparsification approaches~\cite{pagh_tsour,doulion,sparsifiers} can handle edge deletions but the theoretical guarantees on the complexity of the algorithms depend on specific properties of the underlying graph, e.g., the maximum number of triangles an edge is part of. The main contribution in the present work is a novel technique for sampling 2-paths {\em after} the stream has been processed. It is based on the combination of standard 2-path sampling with graph sparsification. The main technical challenge is to show that sampling at random a 2-path in a sparsified graph is (almost) equivalent to sampling at random a 2-path in the original graph. In the course of the analysis, we also obtain combinatorial results about general graphs that might be of independent interest.
\\\\
{\bf Organization of the paper.} In Section~\ref{sec:prel} we give necessary definitions and in Section~\ref{sec:results} we summarize our results. In Section~\ref{sec:upper} we present the new approach, in Section~\ref{sec:code} we describe the algorithm and in Section~\ref{sec:analysis} we prove the main result. In Section~\ref{sec:lower} we provide a lower bound on the space complexity of triangle counting algorithms. We discuss the complexity of different triangle counting algorithms for several real graphs in Section~\ref{sec:dtc_compare}. The paper is concluded in Section~\ref{sec:concl}.

\section{Preliminaries} \label{sec:prel}

{\bf Notation.}
A simple undirected graph without loops is denoted as $G= (V, E)$ with $V = \{1,2,\ldots,n\}$ being a set of vertices and $E$ a set of edges. The edges are provided as a stream of insertions and deletions in arbitrary order. We assume the strict turnstile model where each edge can be deleted only after being inserted. We assume that $n$ is known in advance\footnote{More generally, our results hold when the $n$ vertices come from some arbitrary universe $U$ known in advance.} and that the number of edges cannot exceed $m$. For an edge connecting the vertices $u$ and $v$ we write $(u, v)$ and $u$ and $v$ are the {\em endpoints} of the edge $(u, v)$. Vertex $u$ is {\em neighbor} of $v$ and vice versa and $N(u)$ is the set of $u$'s neighbors. We say that edge $(u, v)$ is {\em isolated} if $|N(u)|=|N(v)|=1$. We consider only edges $(u, v)$ with $u<v$. A {\em 2-path centered at $v$}, $(u, v, w)$, consists of the edges $(u, v)$ and $(v, w)$. A {$k$-clique} in $G$ is a subgraph of $G$ on $k$ vertices $v_1, \ldots, v_k$ such that $(v_i, v_j) \in E$ for all $1\le i < j \le k$. A 3-clique on $u, v, w$ is called a {\em triangle} on $u, v, w$, and is denoted as $\langle u, v, w \rangle$.  We denote by $P_2(v)$ the number of 2-paths centered at a vertex $v$, and $P_2(G)= \sum_{v\in V} P_2(v)$ and $T_3(G)$ the number of 2-paths and number of triangles in $G$, respectively. We will omit $G$ when clear from the context. 

We say that two 2-paths are {\em independent} if they have at most one common vertex.
%A {\em tree} $T = (V, E)$ is a connected graph without cycles. We define a {\em parent-child} relationship on a tree such that each vertex has at most one parent. There exists a single vertex without a parent called the {\em root} of $T$ and vertices with no children are called a {\em leaf}.  A {\em spanning tree} $T_G = (V_T, E_T)$ of $G = (V, E)$ is a tree such that $V_T = V$ and $E_T \subseteq E$. $H = (V_H, E_H)$ is said to be an induced subgraph of $G = (V, E)$ if it has exactly the edges that appear in $G$ over the same vertex set, i.e. $V_H \subseteq V, E_H \subseteq E$ and $(u, v) \in E_H$ if and only if $(u, v) \in E$.
%\\
The {\em transitivity coefficient} of $G$ is $$\alpha(G) = \frac{3T_3}{\sum_{v \in V}{d_v \choose 2}} = \frac{3T_3}{P_2},$$
i.e., the ratio of 2-paths in $G$ contained in a triangle to all 2-paths in $G$. When clear from the context, we will omit $G$.
\\\\
{\bf Hashing.} A family $\mathcal{F}$ of functions from $U$ to a finite set $S$ is {\em $k$-wise independent} if for a function $f: U \rightarrow S$ chosen uniformly at random from $\mathcal{F}$ it holds $$\Pr[f(u_1) = c_1 \wedge f(u_2) = c_2 \wedge \dots \wedge f(u_k) = c_k] = 
{1}/{s^k}$$ for $s = |S|$, distinct $u_i \in U$ and any $c_i \in S$ and $k \in \mathbb{N}$. We will call a function chosen uniformly at random from a $k$-wise independent family {\em $k$-wise independent function} and a function $f: U \rightarrow S$ {\em fully random} if $f$ is $|U|$-wise independent.  We will say that a function $f:U \rightarrow S$ {\em behaves like a fully random function} if for any set of input from $U$, with high probability $f$ has the same probability distribution as a fully random function.   

We will say that an algorithm returns an {\em $(\varepsilon, \delta)$-approximation} of some quantity $q$ if it returns a value $\tilde{q}$ such that $(1-\varepsilon)q \le \tilde{q} \le (1 + \varepsilon)q$ with probability at least $1-\delta$ for every $0 < \varepsilon, \delta < 1$. 
\\\\
{\bf Probability inequalities.}  In the analysis of the algorithm we use the following inequalities. 
\begin{itemize}

%\item {\em Markov's inequality.} Let $X$ be a random variable. Then for every $\lambda > 1$:
%$$
%\Pr[X \geq \lambda \E[X]] \leq \frac{1}{\lambda}
%$$

\item {\em Chebyshev's inequality.} Let $X$ be a random variable and $\lambda>0$. Then

$$
\Pr[|X-\mathbb{E}[X]| \geq \lambda \sigma(X)] \leq \frac{1}{\lambda^2}
$$

\item {\em Chernoff's inequality.} Let $X_1, \ldots, X_\ell$ be $\ell$ independent identically distributed Bernoulli random variables and $\mathbb{E}[X_i] = \mu$. Then for any $\varepsilon > 0$ we have 
$$
\Pr[|\frac{1}{\ell}\sum_{i=1}^\ell X_i - \mu| > \varepsilon\mu] \leq 2e^{-\varepsilon^2 \mu \ell/2}
$$
\end{itemize}

\section{Results} \label{sec:results}
%Our main result is the following. 
The following theorem is our main result.
\begin{thmx} \label{dtc:thm}
Let $G = (V, E)$ be a graph given as a stream of edge insertions and deletions with no isolated edges and vertices, $V=\{1,2,\ldots,n\}$ and $|E| \le m$. Let $P_2$, $T_3$ and $\alpha$ be the number of 2-paths, number of triangles and the transitivity coefficient of $G$, respectively. Let $\varepsilon, \delta \in (0,1)$ be user defined. Assuming fully random hash functions, there exists a one-pass algorithm running in expected space $O(\frac{\sqrt{m}}{\varepsilon^3 \alpha} \log \frac{1}{\delta})$ and $O(\frac{1}{\varepsilon^2 \alpha} \log \frac{1}{\delta})$ processing time per edge. After processing the stream, an $(\varepsilon, \delta)$-approximation of $T_3$ can be computed in expected time $O(\frac{\log n}{\varepsilon^2 \alpha} \log \frac{1}{\delta})$ and worst case time $O(\frac{\log^2 n}{\varepsilon^2 \alpha} \log \frac{1}{\delta})$ with high probability.
\end{thmx}

(For simplicity, we assume that there are no isolated edges in $G$. More generally, the result holds by replacing $n$ with $n_C$, where $n_C$ is the number of vertices in connected components with at least two edges. We assume $m$ and $n$ can be described in $O(1)$ words.)

\begin{table} 
\begin{center} 
%\vspace{-5mm}
\begin{tabular}{ p{4cm} | c | c }
  \hline   
  & Space & Update time\\    
  \hline   
  Ahn et al.\cite{ahn_et_al} & $O(\frac{mn}{\varepsilon^2 T_3} 1 \log \frac{1}{\delta})$ & $O(n \log n)$\\    
  \hline   
  Manjunath et al.~\cite{cycles_counting} & $O(\frac{m^3}{\varepsilon^2 T^2_3}  \log \frac{1}{\delta})$ & $O(\frac{m^3}{\varepsilon^2 T^2_3} \log \frac{1}{\delta})$\\    
  \hline   
  Section~\ref{sec:lower} &  $\Omega(\frac{m^3}{T^2_3})$  &  --- \\
  \hline          
  Section~\ref{sec:upper} &  $O(\frac{\sqrt{m}}{\varepsilon^3 \alpha} \log \frac{1}{\delta})$ & $O(\frac{1}{\varepsilon^2 \alpha} \log \frac{1}{\delta})$\\
  \hline   
\end{tabular} 
 %\vspace{-1mm} 
\caption{Overview of time and space bounds.}
\label{dtc_tab:results}
\end{center}
\end{table}

\begin{table} 
\begin{center} 
%\vspace{-15mm}
\begin{tabular}{ p{2.5cm} | p{4.5cm} | p{3cm} }
  \hline   
  \hspace{9mm} & \hspace{5mm}$n \log n$ & \hspace{5mm}$m^3/T^2_3$\\
  \hline   
  \hspace{8mm}$z < 1/2$ & \hspace{5mm}$T_3 = \omega({C^2}/({n^{2z}\log n}))$ & \hspace{3mm}$T_3 = o(Cn^{2-z})$\\    
  \hline   
  \hspace{4mm}$1/2 < z < 1$ & \hspace{5mm}$T_3 = \omega({C^2}/({n\log n}))$ & \hspace{3mm}$T_3 = o(Cn^{3 - 3z})$\\    
  \hline   
  \hspace{9mm}$z > 1$ &  \hspace{5mm}$T_3 = \omega({C^2}/({n\log n}))$ & \hspace{3mm}$T_3 = o(C)$\\
  \hline          
\end{tabular} 
% \vspace{-1mm}  
\caption{Comparison of the theoretical guarantees for the per edge processing time for varying $z$.}
\label{dtc_tab:time}
\end{center}
%\vspace{-2mm}  
\end{table}  

Before presenting the algorithm, let us compare the above to the bounds in~\cite{ahn_et_al,cycles_counting}. The algorithm in~\cite{ahn_et_al} estimates $T_3$ by applying $\ell_0$ sampling~\cite{L0} to non-empty subgraphs on $3$ vertices. There are $O(mn)$ such subgraphs, thus $O(\frac{mn}{\varepsilon^2 T_3} \log \frac{1}{\delta})$ samples are needed for an $(\varepsilon, \delta)$-approximation. However, each edge insertion or deletion results in the update of $n-2$ non-empty subgraphs on 3 vertices. Using the $\ell_0$ sampling algorithm from~\cite{L0_sampling}, this results in processing time of $O(n \log n)$ per edge. The algorithm by Manjunath et al.~\cite{cycles_counting} (which builds upon the the work of Jowhari and Ghodsi~\cite{iran_triangles}) estimates the number of triangles (and more generally of cycles of fixed length) in streamed graphs by computing complex valued sketches of the stream.  Each of them yields an unbiased estimator of $T_3$. The average of $O(\frac{m^3}{\varepsilon^2 T^2_3} \log \frac{1}{\delta})$ estimators is an $(\varepsilon, \delta)$-approximation of $T_3$. However, each new edge insertion or deletion has to update all estimators, resulting in update time of $O(\frac{m^3}{\varepsilon^2 T^2_3} \log \frac{1}{\delta})$. The algorithm was generalized to counting arbitrary subgraphs of fixed size in~\cite{subgraph_counting}.

The time and space bounds are summarized in Table~\ref{dtc_tab:results}. Comparing our space complexity to the bounds in~\cite{ahn_et_al,cycles_counting}, we see that for several graph classes our algorithm is more time and space efficient. (We ignore $\varepsilon$ and $\delta$ and logarithmic factors in $n$ for the space complexity.) For $d$-regular graphs the processing time per edge is better than $O(n \log n)$ for $T_3 = \omega(d^2/\log n)$, and better than $O(m^3/T^2_3)$ for $T_3 = o(n^2 d)$. Our space bound is better than $O({mn}/{T_3})$  when $d = o(n^{1/3})$, and better than $O({m^3}/{T^2_3})$ for $T_3 = o(n^{3/2}d^{1/2})$. Most real-life graphs exhibit a skewed degree distribution adhering to some form of power law, see for example~\cite{power_law}. Assume vertices are sorted according to their degree in decreasing order such that the $i$th vertex has degree $C/i^z$ for some $C \le n$, and constant $z>0$, i.e., we have Zipfian distribution with parameter~$z$. It holds $\sum_{i=1}^ni^{-z} = O(n^{1-z})$ for $z<1$ and $\sum_{i=1}^ni^{-z} = O(1)$ for $z>1$. Table~\ref{dtc_tab:time} summarizes for which values of $T_3$ our algorithm achieves faster processing time than~\cite{ahn_et_al,cycles_counting}, and Table~\ref{dtc_tab:space} -- for which values of $C$ our algorithm is more space-efficient than~\cite{ahn_et_al}, and for which values of $T_3$ -- more space-efficient than~\cite{cycles_counting}.

However, the above values are for arbitrary graphs adhering to a certain degree distribution. We consider the main advantage of the new algorithm to be that it achieves constant processing time per edge for graphs with constant transitivity coefficient. This is a common assumption for real-life networks, see for instance~\cite{albert,buriol_et_al_1}.  Note that fast update is essential for real life applications. Consider for example the Facebook graph. In May 2011, for less than eight years existence, there were about 69 billion friendship links~\cite{facebook}. This means an average of above 300 new links per second, without counting deletions and peak hours.
 
In Section~\ref{sec:dtc_compare} we compare the theoretical guarantees for several real life graphs. While such a comparison is far from being a rigorous experimental evaluation, it clearly indicates that the processing time per edge in~\cite{ahn_et_al,cycles_counting} is prohibitively large and the assumption that the transitivity coefficient is constant is justified. Also, for graphs with a relatively small number of triangles our algorithm is much more space-efficient.

\begin{table}
%\vspace{-6mm}
%\begin{adjustwidth}{-1.6cm}{}
\begin{center} 
\begin{tabular}{ p{3cm} | p{5cm} | p{5cm} }
  \hline   
  \hspace{9mm} & \hspace{5mm}$mn/T_3$ & \hspace{3mm}$m^3/T^2_3$\\    
  \hline   
  \hspace{8mm}$z < 1/2$ & \hspace{5mm}$C = o(n^{1/3 + z})$ & \hspace{3mm}$T_3 = o(C^{1/2} n^{\frac{3-z}{2}}))$\\    
  \hline   
  \hspace{4mm}$1/2 < z < 1$ & \hspace{5mm}$C = o(n^{1-z/3})$ & \hspace{3mm}$T_3 = o(C^{1/2}n^{\frac{5 - 5z}{2}})$\\    
  \hline   
  \hspace{9mm}$z > 1$ &  \hspace{5mm}$C = o(n^{2/3})$ & \hspace{3mm}$T_3 = o(C^{1/2})$\\
  \hline          
\end{tabular}  
 %\vspace{-1mm} 
\caption{Comparison of the theoretical guarantees for the space usage for varying $z$.}
 \label{dtc_tab:space}
\end{center}
%\vspace{-15mm}
%\end{adjustwidth}
\end{table}

\section{Our algorithm}\label{sec:upper}
The main idea behind our algorithm is to design a new sampling technique for dynamic graph streams. It exploits a combination of the algorithms by Buriol et al.~\cite{buriol_et_al_1} for the incidence stream model, and the Doulion algorithm~\cite{doulion} and its improvement~\cite{pagh_tsour}. Let us briefly describe the approaches. 
\\\\
{\bf The Buriol et al. algorithm for incidence list streams.}
Assume we know the total number of 2-paths in $G$. One chooses at random one of them, say $(u, v, w)$, and checks whether the edge $(u,w)$ appears later in the stream. For a triangle $\langle u, v, w \rangle$ the three 2-paths $(u, v, w)$, $(w, u, v)$, $(v, w, u)$ appear in the incidence list stream, thus the probability that we sample a triangle is exactly $\alpha$. One chooses independently at random $K$ 2-paths and using standard techniques shows that for $K = O(\frac{1}{\varepsilon^2 \alpha}\log \frac{1}{\delta})$ we compute an $(\varepsilon, \delta)$-approximation of $\alpha(G)$. One can get rid of the assumption that the number of 2-paths is known in advance by running $O(\log n)$ copies of the algorithm in parallel, each guessing the right value. The reader is referred to the original work for more details. For incidence streams, the number of 2-paths in $G$ can be computed exactly by updating a single counter, thus if $\tilde{\alpha}$ is an $(\varepsilon, \delta)$-approximation of the transitivity coefficient then $\tilde{T_3} = \tilde{\alpha}P_2$ is an $(\varepsilon, \delta)$-approximation of $T_3$. 
\\\\
{\bf Doulion and monochromatic sampling.}
The Doulion algorithm~\cite{doulion} is a simple and intuitive sparsification approach. Each edge is sampled independently with probability $p$ and added to a sparsified graph $G_S$. We expect $pm$ edges to be sampled and a triangle survives in $G_S$ with probability $p^3$, thus multiplying the number of triangles in $G_S$ by $1/p^3$ we obtain an estimate of $T_3$. The algorithm was improved in~\cite{pagh_tsour} by using {\em monochromatic sampling}. Instead of throwing a biased coin for each edge, we uniformly at random color each vertex with one of $1/p$ colors. Then we keep an edge in the sparsified graph iff its endpoints have the same color. A triangle survives in $G_S$ with probability $p^2$. It is shown that for a fully random coloring the variance of the estimator is better than in Doulion. However, in both algorithms it depends on the maximum number of triangles an edge is part of, and one might need constant sampling probability in order to obtain an $(\varepsilon, \delta)$-approximation on $T_3$. The algorithm can be applied to dynamic streams because one counts the number of triangles in the sparsified graph after all edges have been processed. However, it can be expensive to obtain an estimate since the exact number of triangles in $G_S$ is required.  %The decreased exponent also results in a better space complexity and the sampling probability required for an $(\varepsilon, \delta)$-approximation is improved to $p = \Omega(\frac{\Delta}{\varepsilon^2 T_3}\log \frac{1}{\delta})$. As one can see, the main drawback of the approach is that it is not suitable for graphs where certain edges are ``heavy", i.e., edges contained in many triangles. %However, some examples where the approach fails have also been reported~\cite{?}(koreicite).   
\\\\
{\bf Combining the above approaches.} 
The basic idea behind the new algorithm is to use the estimator of Buriol et al. for the incidence stream model: (i) estimate the transitivity coefficient $\alpha(G)$ by choosing a sufficiently large number of 2-paths at random and check which of them are part of a triangle, and (ii) estimate the number of 2-paths $P_2$ in the graph. We first observe that estimating $P_2$ in dynamic graph streams can be reduced to second moment estimation of streams of items in the turnstile model, see e.g.~\cite{f2_est_opt}. For (i), we will estimate $\alpha(G)$ by adjusting the monochromatic sampling approach. Its main advantage compared to the sampling of edges separately is that if we have sampled the 2-path $(u, v, w)$, then we must also have sampled the edge $(u, w)$, if existent. So, the idea is to use monochromatic sampling and then in the sparsified graph to pick up at random a 2-path and check whether it is part of a triangle. Instead of random coloring of the vertices, we will use a suitably defined hash function and we will choose a sampling probability guaranteeing that for a graph with no isolated edges (or rather a small number of isolated edges) the sparsified graph will contain a sufficiently big number of 2-paths. A 2-path in the sparsified graph picked up at random, will then be used to estimate $\alpha(G)$. Thus, unlike in~\cite{buriol_et_al_1}, we sample {\em after} the stream has been processed and this allows to handle edge deletions. The main technical obstacles are to analyze the required sampling probability $p$ and to show that this sampling approach indeed provides an unbiased estimator of $\alpha(G)$.  We will obtain bounds on $p$ and show that even if the estimator might be biased, the bias can be made arbitrarily small and one can still achieve an $(\varepsilon, \delta)$-approximation of $\alpha(G)$. Also, we present an implementation for storing a sparsified graph $G_S$ such that each edge is added or deleted in constant time and a random 2-path in $G_S$, if existent, can be picked up without explicitly considering all 2-paths in $G_S$.   %Then it is easy to extend the approach to estimating the number of triangles.     

\subsection{Algorithm details} \label{sec:code}

Pseudocode description of the algorithm is given in Figure~\ref{fig:alg}. We assume that the graph is given as a stream $\mathcal{S}$ of pairs $((u, v), \$$$)$,
 where $(u,v) \in E$ and $\$$ $ \in \{+,-\}$
  with the obvious meaning that the edge $(u, v)$ is inserted or deleted from $G$. In {\sc EstimateNumberOfTwoPaths} each incoming pair $((u, v), \$$$)$ is treated as the insertion, respectively deletion, of two items $u$ and $v$, and these update a second moment estimator $SME$, working as a blackbox algorithm. We refer to the proof of Lemma~\ref{lem:second_moment_est} for more details. In {\sc SparsifyGraph} we assume access to a fully random coloring hash function $f: V \rightarrow C$. Each edge $(u,v)$ is inserted/deleted to/from a sparsified graph $G_S$ iff $f(u) = f(v)$. At the end $G_S$ consists of all monochromatic edges that have not been deleted. %We assume an upper bound on the space {\sc SparsifyGraph} can use and consider only sparsified graphs with at most $t$ monochromatic edges. 
  In {\sc EstimateNumberOfTriangles} we run in parallel the algorithm estimating $P_2$ and $K$ copies of {\sc SampleRandom2Path}. For each $G^i_S$, $1 \le i \le K$, with at least $s$ pairwise independent 2-paths %and at most $t$ edges, 
  we choose at random a 2-path and check whether it is a triangle. (Note that we require the existence of $s$ {\em pairwise independent} 2-paths but we choose a 2-path at random from {\em all} 2-paths in $G_S$.) The ratio of triangles to all sampled 2-paths and the estimate of $P_2$ are then used to estimate $T_3$. In the next section we obtain bounds on the user defined parameters $C, K$ and $s$. In Lemma~\ref{lem:sampling} we present en efficient implementation of $G_S$ that guarantees constant time updates and allows the sampling of a random 2-path in expected time $O(\log n)$ and worst case time $O(\log^2 n)$ with high probability. 

%\begin{figure}
%{\sc EstimateNumberOfTwoPaths}
%\algsetup{indent=2em}
%\begin{algorithmic}[1]
%\REQUIRE stream of edge deletions and insertions $\mathcal{S}$, algorithm $SME$ estimating the second moment items streams
%\STATE $m=0$
%\FOR{each $((u, v), \$)$ in $\mathcal{S}$}
%\IF{$\$ = +$}
%\STATE $m= m+1$
%\STATE $SME.update(u, 1)$, $SME.update(v, 1)$  
%\ELSE
%\STATE $m=m-1$
%\STATE $SME.update(u, -1)$, $SME.update(v, -1)$  
%\ENDIF
%\ENDFOR
%\RETURN $SME.estimate/2 - m$
%\end{algorithmic}
%\caption{The algorithm for estimating the number of 2-paths in $G$.}  \label{fig:sme} 
%\end{figure} 

%\bigskip

%\paragraph{Estimating the number of triangles}

\begin{figure}
\vspace{-5mm}
{\sc EstimateNumberOfTwoPaths}
\algsetup{indent=2em}
\begin{algorithmic}[1]
\REQUIRE stream of edge deletions and insertions $\mathcal{S}$, algorithm $SME$ estimating the second moment items streams
\STATE $m=0$
\FOR{each $((u, v), \$$$)$ in $\mathcal{S}$}
\IF{$\$$$ = +$}
\STATE $m= m+1$
\STATE $SME.update(u, 1)$, $SME.update(v, 1)$  
\ELSE
\STATE $m=m-1$
\STATE $SME.update(u, -1)$, $SME.update(v, -1)$  
\ENDIF
\ENDFOR
\RETURN $SME.estimate/2 - m$
\end{algorithmic}
%\caption{The algorithm for estimating the number of 2-paths in $G$.}  \label{fig:sme} 

\bigskip

{\sc SparsifyGraph}
\algsetup{indent=2em}
\begin{algorithmic}[1]
\REQUIRE stream of edge deletions and insertions $\mathcal{S}$, coloring function $f: V \rightarrow C$%, threshold $t$
\medskip
\STATE $G_S = \emptyset$
\FOR{each $((u, v), \$$$) \in \mathcal{S}$}
\IF{$f(u) = f(v)$}
\IF{$\$$$ = +$}
\STATE  $G_S = G_S \cup (u, v)$.
%\IF{$E(G_S) > t$}
%\STATE Return $\emptyset$.
%\ENDIF
\ELSE 
\STATE $G_S = G_S \backslash (u, v)$.
\ENDIF
\ENDIF
\ENDFOR
\STATE Return $G_S$.
\end{algorithmic}

\bigskip

{\sc SampleRandom2Path}
\algsetup{indent=2em}
\begin{algorithmic}[1]
\REQUIRE sparsified graph $G_S$
\medskip
\STATE choose at random a 2-path $(u, v, w)$ in $G_S$
\IF{the vertices $\{u,v,w\}$ form a triangle}
\RETURN 1
\ELSE 
\RETURN 0
\ENDIF
\end{algorithmic}
%\caption{A single estimator.} 

\bigskip

{\sc EstimateNumberOfTriangles}
\algsetup{indent=2em}
\begin{algorithmic}[1]
\REQUIRE streamed graph $\mathcal{S}$, set of $K$ independent fully random coloring functions $\mathcal{F}$, algorithm $SME$ estimating the second moment of streams of items, threshold $s$% and $t$
\medskip 
\STATE run in parallel {\sc EstimateNumberOfTwoPaths}$(\mathcal{S}, SME)$ and let $\tilde{P_2}$ be the returned estimate
\STATE run in parallel $K$ copies of {\sc SparsifyGraph}$(S, f_i)$, $f_i \in \mathcal{F}$ 
\STATE $\ell = 0$
\FOR{each $G^i_S$ with at least $s$ pairwise independent 2-paths}
\STATE $X += $ {\sc SampleRandom2Path}$(G^i_S)$
\STATE $\ell += 1$
\ENDFOR
\STATE $\tilde{\alpha} = X/\ell$
\RETURN $\frac{\tilde{\alpha} \tilde{P_2}}{3}$
\end{algorithmic}
\caption{Estimating the number of 2-paths in $G$, the transitivity coefficient and the number of triangles.} 

\label{fig:alg} 
\end{figure}

%%%%%%%%%%%%%%%%%%%%%%%%%%%%%%%%%%%%%%%%%%%%%%%%%%%%%%%%%%%%%%%%%%%%%%%%%%%%%%%%%%%%%%%%%%%%%%%%%%%%
%%%%%%%%%%%%%%%%%%%%%%%%%%%%%%%%%%%%%%%%%%%%%%%%%%%%%%%%%%%%%%%%%%%%%%%%%%%%%%%%%%%%%%%%%%%%%%%%%%%%
%%%%%%%%%%%%%%%%%%%%%%%%%%%%%%%%%%%%%% ANALYSIS %%%%%%%%%%%%%%%%%%%%%%%%%%%%%%%%%%%%%%%%%%%%%%%%%%%%
%%%%%%%%%%%%%%%%%%%%%%%%%%%%%%%%%%%%%%%%%%%%%%%%%%%%%%%%%%%%%%%%%%%%%%%%%%%%%%%%%%%%%%%%%%%%%%%%%%%%
%%%%%%%%%%%%%%%%%%%%%%%%%%%%%%%%%%%%%%%%%%%%%%%%%%%%%%%%%%%%%%%%%%%%%%%%%%%%%%%%%%%%%%%%%%%%%%%%%%%%

\subsection{Theoretical analysis} \label{sec:analysis}

We will prove the main result in several lemmas. The first lemma provides an estimate of $P_2$ using an estimator for the second frequency moment of data streams~\cite{f2_est_opt}.
\begin{lemmx} \label{lem:second_moment_est}
Let $G$ be a graph with no isolated edges given as a stream of edge insertions and deletions. There exists an algorithm returning an $(\varepsilon, \delta)$-approximation of the number of 2-paths in $G$ in one pass over the stream of edges which needs $O(\frac{1}{\varepsilon^2} \log \frac{1}{\delta})$ space and $O(\log \frac{1}{\delta})$ processing time per edge.
\end{lemmx}
\begin{proof}
We show that {\sc EstimateNumberOfTwoPaths} in Figure~\ref{fig:alg} returns an $(\varepsilon, \delta)$-approximation of the number of 2-paths in $G$. 
We reduce the problem of computing the number of 2-paths in dynamic graph streams to the problem of estimating the second frequency moment in streams of items in the turnstile model. By associating vertices with items and treating each incoming pair $((u, v), \$)$ %$
 as the insertion, respectively deletion, of two new items $u$ and $v$, it is a simple observation that the second moment of the so defined stream of items corresponds to $F_2 = \sum_{v \in V} d_v^2$. For the number of 2-paths in $G$ we have $$P_2 = \sum_{v \in V} {d_v \choose 2} = (\sum_{v \in V} d_v^2)/2 - m.$$ Since $G$ contains no isolated edges, for each edge $(u, v)$ we can assume that at least one of its endpoints has degree more than one, w.l.o.g. $d_u \ge 2$. Thus, $(u, v)$ must be part of at least one 2-path, namely a 2-path centered at $u$. Each edge can be assigned thus to at least one 2-path and this implies a lower bound on the number of 2-paths in $G$, $\sum_{v \in V} {d_v \choose 2} \ge m/2$, thus $$F_2 = 2\sum_{v \in V} {d_v \choose 2} + 2m \ge 3m.$$
Let $c \ge 1$ be some constant. Assume that we have computed an $(1 \pm \varepsilon/c)$-approximation of $F_2$ and we have the exact value of $m$. Then the over- and underestimation returned by {\sc EstimateNumberOfTwoPaths} is bounded by $(\varepsilon/c) F_2$. We want to show that $(\varepsilon/c) F_2 \le \varepsilon P_2$. It holds $$(\varepsilon/c) F_2 \le \varepsilon P_2 = \varepsilon (F_2/2 - m)$$ $$F_2/c \le F_2/2 - m $$ $$F_2(c-2) \ge 2cm$$ Since $F_2 \ge 3m$, for $c\ge 6$ we obtain an $(1\pm \varepsilon)$-approximation of $P_2$. 

As for the complexity of {\sc EstimateNumberOfTwoPaths}, we observe that recording the exact value of $m$ requires constant processing time per edge and $O(\log n)$ bits. For $SME$, we use the algorithm from~\cite{f2_est_opt} for the estimation of the second moment, which computes an $(\varepsilon, \delta)$-approximation of $\sum_{v \in V} d_v^2$ with the claimed time and space complexity.  
\end{proof}
\\\\
The next two lemmas show a lower bound of $\Omega(m)$ on the number of pairwise independent 2-paths in a graph without isolated edges. The result is needed in order to obtain bounds on the required sampling probability. 
We first show that in order to obtain a lower bound for general graphs it is sufficient to
consider bipartite connected graphs.
This is true because every connected graph contains a cut with at
least~$m/2$ edges such that the bipartite graph between the two vertex
subsets is connected. We prove this first.

\begin{lemmx}\label{lem:reducteToBipartite}
  Let~$G=(V,E)$ be an arbitrary connected graph with~$|E|=m$.
  There exists a bipartition~$(U,W)$ of~$V$
  such that the bipartite graph~$B:=(V,E')$ with $E':=E\cap\{\{u,w\}\mid u\in U, w \in
  W\}$ is connected and~$|E'|\ge m/2$.
\end{lemmx}
\begin{proof}
  Let~$G=(V,E)$ be a connected graph. We first show that there exists
  $(U,W)$ such that~$B$ contains at least~$m/2$ edges.
  To this end, initialize~$U:=\emptyset$ and~$W:=V$.
  Now, as long as there either exists a vertex~$v\in W$ with~$|N_G(v)\cap W| >
  |N_G(v)\cap U|$ or a vertex~$v\in U$ with~$|N_G(v)\cap U| >
  |N_G(v)\cap W|$, we exchange this vertex~$v$, that is, in the first
  case, we add~$v$ to~$U$ and delete it from~$W$ and vice versa in the
  second case. Clearly, this increases the number of edges between~$U$
  and~$W$; hence, this procedure terminates after a finite number of steps.
  Moreover, after termination, it holds~$|N_G(v)\cap W| \le
  |N_G(v)\cap U|$ for each~$v\in W$ and $|N_G(v)\cap U| \le
  |N_G(v)\cap W|$ for each~$v\in U$. Thus, there are at least~$m/2$
  edges in~$B$.
  
  Assume now that~$B$ is not connected and let~$B_1,\ldots,B_c$ denote
  the~$c\ge 2$ connected components of~$B$ and let~$U_i\subseteq U$,
  $W_i\subseteq W$ be the bipartition of~$B_i$. We simply ``swap'' an
  arbitrary component, say~$B_1$, that is, we delete the
  vertices~$U_i$ from~$U$ and add them to~$W$ and delete the
  vertices~$W_i$ from~$W$ adding them to~$U$. Clearly,
  the edges between~$U_i$ and~$W_i$ remain in~$B$. Additionally,
  the edges between~$U_i$ and~$U\setminus U_i$ as well as the edges between~$W_i$
  and~$W\setminus W_i$ are added to~$B$. Hence, the number of edges in~$B$
  increased. Moreover, since~$G$ is connected, the component~$B_1$ is
  now connected to some other component in $B$ and thus the number of
  connected components decreased. This proves the claim.
\end{proof}
\\\\
We also need, for the case where there are few edges in the graph, the
following lower bound.
\begin{lemmx}\label{lem:trees}
 Let $G=(V,E)$ be a connected graph.
 The number of \emph{independent 2-paths} in $G$ is at least $|V|/2-1$.
\end{lemmx}
\begin{proof}
 This can easily be obtained by taking a rooted spanning tree $T$ of $G$. 
 Assuming $|V|\geq 3$, take any leaf $u$ of $T$ with maximum depth (distance to the root). Let $v$ be the parent of $u$. 
 If $v$ has another child $w$, then $w$ is a leaf of $T$. Select the 2-path $\{\{u,v\},\{v,w\}\}$, remove vertices $u$ and $w$ from $V$, and start again.
 Otherwise, $v$ has a parent $w$: Similarly, select the 2-path $\{\{u,v\},\{v,w\}\}$, remove vertices $u$ and $v$ from $V$, and start again. 
 The set of 2-paths thus selected is independent (each one uses 2 new
 vertices) and contains at least $(m-1)/2=(|V|-2)/2$ 2-paths.
\end{proof}
\\\\
We now prove the lower bound for bipartite connected graphs.
\begin{lemmx} \label{lem:bipartite}
 Let $G=(V,E)$ with~$|E|=m$ be a connected bipartite graph.
 The number of \emph{independent 2-paths} in~$G$ is at least~$\lfloor m/9 \rfloor$.
\end{lemmx}
\begin{proof}
 We say that a set of independent 2-paths $\mathcal P$ is \emph{maximal} if for any 2-path $P$ of $G$ not yet in $\mathcal P$, $\mathcal P\cup\{P\}$ is not independent.  
 We first prove the following property by induction on $|E|$:
 
 {\em 
   In any bipartite graph $G=(V,E)$, if $\mathcal P$ is a maximal set of independent 2-paths, then $$ |\mathcal P| \geq  \frac m6  - \frac{|V|}4. $$
 }
     
 Let $F$ be the set of edges which are not used in any 2-path of $\mathcal P$. Note that $|\mathcal P|=\frac{m-|F|}2$. We consider the subgraph $H:=(V,F)$ of $G$.
 First, consider the easy case where all vertices have degree at most $1$ in $H$. It follows that $|F|\leq \frac{|V|}2$. Thus,  $$|\mathcal P|\geq \frac{m-\frac{|V|}2}2 \geq \frac m6-\frac{|V|}4.$$
 Otherwise, pick a vertex $u$ having maximum degree in $H$, write $k$ for its degree ($k\geq 2$), 
 and $N(u)$ for the set of its $k$ neighbors in $H$.
 Consider any two distinct vertices $v, v'\in N(u)$. Then $\{\{v,u\},\{u,v'\}\}$ is a 2-path of~$G$. 
 Since $\mathcal P$ is maximal, this 2-path must be in conflict with some selected 2-path~$P_{v,v'}\in \mathcal P$.
 So, $P_{v,v'}$ covers two vertices among $\{v,u,v'\}$. Since the
 graph is bipartite, and $P_{v,v'}$ does not use the edges $\{u,v\},\{u,v'\}\in F$,
 it follows that $P_{v,v'}$ must cover both vertices $v$ and $v'$.
 Let $F(u)\subseteq F$ be the subset of edges of $F$ which are incident to a vertex of $N(u)$. 
 Note that $|F(u)|\leq k^2$. Indeed, $N(u)$ contains $k$ vertices, all of them having degree at most $k$ in $H$.
 Let $P(u)\subseteq E\setminus F$ be the subset of $E$ containing all
 edges used by a 2-path $P_{v,v'}\in \mathcal P$ for any pair~$v,v'\in N(u)$, then  $|P(u)|=2\binom{k}{2}=k(k-1)$. 
 Indeed, there are $\binom{k}{2}$ pairs $v,v'\in N(u)$, and each of them yields a 
 unique 2-path $P_{v,v'}$ containing two edges (two 2-paths cannot
 share any edge since they are independent).
 Now, we define $$E':=E\setminus (P(u)\cup F(u)) \text{ and }\mathcal P':=\mathcal P\setminus \{P_{v,v'} \mid v\neq v'\in N(u)\},$$ 
 and claim that $\mathcal P'$ is a maximal set of 2-paths of the
 subgraph $G':=(V,E')$. Assume towards a contradiction that there is a
 2-path $P^*$ in $G'$ which 
 is not in conflict with any 2-path of $\mathcal P'$. Since $G'$ is a subgraph of $G$, and $\mathcal P$ is maximal, it follows that $P^*$ is in conflict 
 with some $P_{v,v'}\in \mathcal P \setminus \mathcal P'$. Thus, $P^*$
 contains at least one of the two vertices $v,v'$, and must contain an
 edge in $F(u)\subseteq E\setminus E'$;
 a contradiction. Hence, $\mathcal P'$ is maximal for the bipartite graph $G'$ and, by induction, we have $$|\mathcal P'|\geq \frac{|E'|}6-\frac{|V|}4.$$
 Putting everything together, we have
 \[|\mathcal P'| = |\mathcal P|-\binom{k}{2} \text{ and } |E'| \geq |E|-k(k-1)-k^2,\]
which yields
 \begin{align*}
   |\mathcal P| &= |\mathcal P'| + \binom{k}{2} \geq \frac{|E'|}{6}-\frac{|V|}{4}+\binom{k}{2}\\
   &\ge \frac{|E|-k(k-1)-k^2}6 - \frac{|V|}4 + \binom{k}{2}\\
   &= \frac{|E|}6 - \frac{|V|}4 + \frac{2k(k-1)-k^2}6\\
   &\geq  \frac{|E|}6 - \frac{|V|}4 \text{\quad (since  $2k(k-1) \geq {k^2}$ for all $k\geq 2$).}    
 \end{align*}
 This completes the proof of the induction property.  

 To prove the lemma, it remains to distinguish between the following two cases:
 \begin{itemize}
  \item If $m=|E|\geq \frac{9|V|}2$, then with the induction property above, 
    any maximal set of independent 2-paths has size at
    least $$\frac{m}6-\frac{|V|}4 \ge \frac{m}6 - \frac14\cdot\frac{2m}{9}=\frac {3m-m}{18} \geq\left\lfloor \frac m9 \right\rfloor. $$
  \item If $m=|E|<  \frac{9|V|}2$, then with Lemma~\ref{lem:trees},
    there exists a set of independent 2-paths with size at least
    $\frac{|V|}2-1 \ge \lceil\frac m9 - 1\rceil = \left\lfloor \frac m9 \right\rfloor$.
 \end{itemize}
\end{proof}
\\\\
Combining Lemmas~\ref{lem:reducteToBipartite} and \ref{lem:bipartite} we obtain that any connected graph has at least $\left\lfloor \frac m{18} \right\rfloor$ independent 2-paths. For dense graphs, Lemma~\ref{lem:bipartite} gives a lower bound of $\frac m{12}-\frac {|V|}4\sim \frac m{12}$. This yields the following result:

\begin{thmx} \label{thm:lb2paths}
  Let $G=(V,E)$ with~$|E|=m$ be a connected graph. The number of
  \emph{independent 2-paths} in~$G$ is in~$\Omega(m)$.
\end{thmx}
Next we obtain bounds on the sampling probability such that we can guarantee there are sufficiently many pairwise independent 2-paths in $G_S$. As we show later, this is needed to guarantee that {\sc SampleRandom2Path} will return an almost unbiased estimator of the transitivity coefficient. The events for two 2-paths being monochromatic are independent, thus the next lemma follows from Theorem~\ref{thm:lb2paths} and Chebyshev's inequality. Note that we still don't need the coloring function $f$ to be fully random.

\begin{lemmx} \label{lem:sampl_prob}
Let $f$ be 6-wise independent and $p \ge \frac{5}{\varepsilon\sqrt{b}}$ where $b$ is the number of independent 2-paths in $G$ and $\varepsilon \in (0,1]$. Then with probability at least 3/4 {\sc SparsifyGraph} returns $G_S$ such that there are at least $18/\varepsilon^2$ independent 2-paths in $G_S$.
\end{lemmx}
\begin{proof}
Let $D_2$ be a set of pairwise independent 2-paths in $G$, $|D_2| = b$. Clearly, each 2-path in $D_2$ is monochromatic with probability $p^2$. Consider two pairwise independent 2-paths $(u_1, v_1, w_1)$ and $(u_2, v_2, w_2)$. If they do not share a vertex, then since $f$ is 6-wise independent, the two of them are monochromatic with probability $p^4$. Otherwise, assume w.l.o.g. $u_1 = u_2$. Under the assumption that $f(u_1) = c$ for some $c \in C$, we have that evaluating $f$ on $v_i, w_i$, $i=1,2$, is 5-wise independent. Thus, $\Pr[f(v_i)=f(w_i) = c] = p^4$. The events of sampling the two 2-paths are thus independent. We introduce an indicator random variable $X_{(u,v,w)}$ for each 2-path $(u, v, w) \in D_2$ denoting whether $(u,v,w) \in G_S$ and set $X = \sum_{(u, v, w) \in D_2} X_{(u,v,w)}$.  We have $\E[X] \ge p^2 b = 25/\varepsilon^2$ and since the events that any two 2-paths in $D_2$ are sampled are independent, we have $V[X] \le \E[X]$. From Chebyshev's inequality with some algebra we then obtain that with probability at least $3/4$ we have at least $18/\varepsilon^2$ monochromatic pairwise independent 2-paths in $G_S$. 
\end{proof}

\begin{lemmx} \label{lem:estimate}
Assume we run {\sc EstimateNumberOfTriangles} with $s = 18/\varepsilon^2$ and let $X$ be the value returned by {\sc SampleRandom2Path}. Then $(1-\varepsilon)\alpha \le \E[X] \le (1+\varepsilon)\alpha$. 
\end{lemmx}
\begin{proof}
We analyze how much differs the probability between 2-paths to be selected by {\sc SampleRandom2Path}.
Consider a given 2-path $(u, v, w)$. It will be sampled if the following three events occur:
\begin{enumerate}
\item $(u, v, w)$ is monochromatic, i.e., it is in the sparsified graph $G_S$.
\item There are $i \ge 18/\varepsilon^2$ pairwise independent 2-paths in $G_S$.
\item $(u, v, w)$ is selected by {\sc SampleRandom2Path}. 
\end{enumerate}
The first event occurs with probability $p^2$. Since $f$ is fully random, the condition that $(u,v,w)$ is monochromatic does not alter the probability for any 2-path independent from $(u,v,w)$ to be also monochromatic. The probability to be in $G_S$ changes only for 2-paths containing two vertices from $\{u,v,w\}$, which in turn changes the number of 2-paths in $G_S$ and thus the probability for $(u,v,w)$ to be picked up by {\sc SampleRandom2Path}. In the following we denote by $p_{G_S}$ the probability that a given 2-path is monochromatic and there are at least $18/\varepsilon^2$ pairwise independent 2-paths in $G_S$, note that $p_{G_S}$ is equal for all 2-paths. 

Consider a fixed coloring to $V\backslash \{u,v,w\}$. We analyze the difference in the number of monochromatic 2-paths depending whether $f(u)=f(v)=f(w)$ or not. There are two types of 2-paths that can become monochromatic conditioning on $f(u)=f(v)=f(w)$: either (i) 2-paths with two endpoints in $\{u,v,w\}$ centered at some $\{u,v,w\}$, or (ii) 2-paths with two vertices in $\{u,v,w\}$ centerer at a vertex $x \notin \{u,v,w\}$. For the first case assume w.l.o.g. there is a 2-path $(u,v,w) \in G_S$  centered at $v$ and let $d^{f(v)}_{v} = |\{z \in N(v)\backslash \{u, w\}: f(z)=f(v)\}|$, i.e., $d^{f(v)}_{v}$ is the number of $v$'s neighbors different from $u$ and $w$, having the same color as $v$. Thus, the number of monochromatic 2-paths centered at $v$ varies by $2d^{f(v)}_v$ conditioning on the assumption that $f(u)=f(v)=f(w)$. The same reasoning applies also to the 2-paths centered at $u$ and $w$. For the second case consider the vertices $u$ and $v$. Conditioning on $f(u)=f(w)$, we additionally add to $G_S$ 2-paths $(u,x_i,w)$ for which $f(x_i) = f(u)=f(w)$ and $x_i \in N(u) \cap N(w)$. The number of such 2-paths is at most $\min(d^{f(u)}_u, d^{f(w)}_w)$.  The same reasoning applies to any pair of vertices from $\{u,v,w\}$. Therefore, depending on whether $f(u)=f(v)=f(w)$ or not, the number of monochromatic 2-paths with at least two vertices from $\{u,v,w\}$ and centered at a vertex from $\{u,v,w\}$ varies between $$\sum_{y \in \{u,v,w\}}{d^{f(y)}_y \choose 2}  \hspace{6mm} \text{ and } \sum_{y \in \{u,v,w\}}{d^{f(y)}_y \choose 2} + 3d^{f(y)}_y.$$  
\\
Set $k={18}/{\varepsilon^2}$. We will bound the sampling probability range for a 2-path $(u, v, w)$. Let $\mathcal{C}$ be a coloring that makes $(u, v, w)$ monochromatic.  Assume that with probability $p_i$, $i-1$ 2-paths are colored monochromatic by $\mathcal{C}$. Under the assumption there are $i \ge k$ 2-paths in $G_S$ and following the above discussion about the number of 2-paths with at least two vertices from $\{u, v, w\}$, we see that the number of monochromatic 2-paths can vary between  $i$ and $i+ 3\sqrt{2i}$. Thus, the probability for $(u, v, w)$  to be sampled varies between $$p_{G_S} \sum_{i \ge k} \frac{p_i}{i} \hspace{6mm} \text{ and } \hspace{6mm} p_{G_S} \sum_{i \ge k} \frac{p_i}{i + 3\sqrt{2i}}.$$ We assume $G_S$ contains at least $k$ 2-paths, thus $\sum_{i \ge k} p_i = 1$ and there exists $r \ge k, r \in \mathbb{R}$ such that $\sum_{i \ge k} p_ii^{-1} = 1/r$. Thus we bound $$\sum_{i \ge k} \frac{p_i}{i + 3\sqrt{2i}} =  \sum_{i \ge k} \frac{p_i}{i(1 + 3\sqrt{{2}/{i})}} \ge \frac{1}{1 + 3\sqrt{{2}/{k}}}\sum_{i \ge k} \frac{p_i}{i} = \frac{1}{r(1 + 3\sqrt{2/k})}.$$

Assume first the extreme case that 2-paths which are not part of a  triangle are sampled with probability $\frac{1}{r}$ and 2-paths part of a triangle with probability $\frac{1}{(1 + \varepsilon)r}$.  
We have $X = \sum_{(u,v,w) \in P_2} I_{(u,v,w)}$, where $I_{(u,v,w)}$ is an indicator random variable denoting whether $(u,v,w)$ is part of a triangle.  Thus  $$\E[X] \ge \frac{p_{G_S}3T_3}{(1 + \varepsilon)r} \frac{r}{p_{G_S}P_2} = \frac{\alpha}{1+\varepsilon} \ge (1-\varepsilon)\alpha.$$ On the other extreme, assuming that we select 2-paths part of triangles with probability $\frac{1}{r}$ and 2-paths not part of a triangle with probability $\frac{1}{r(1 + \varepsilon)}$, using similar reasoning we obtain $\E[X] \le (1+\varepsilon)\alpha$.
\end{proof}
\\\\
Applying a variation of rejection sampling, in the next lemma we show how to store a sparsified graph $G_S$ such that we efficiently sample a 2-path uniformly at random and $G_S$ is updated in constant time. 
\begin{lemmx} \label{lem:sampling}
Let $G_S = (V, E_S)$ be a sparsified graph over $m'$ monochromatic edges. There exists an implementation of $G_S$ in space $O(m')$ such that an edge can be inserted to or deleted from $G_S$ in constant time with high probability. A random 2-path, if existent, can be selected from $G_S$ in expected time $O(\log n)$ and $O(\log^2 n)$ time with high probability.
\end{lemmx}
\begin{proof}
We implement $G_S$ as follows. Let $\Delta \le n$ be the maximum vertex degree in $G_S$. We maintain a hash table $H_i$ for all vertices $v \in G_S$ with $P_2(v) \in \{2^i, 2^i+1, \ldots, 2^{i+1}-1\}$, $0 \le i \le 2 \log \Delta$, i.e., there are between $2^i$ and $2^{i+1}-1$ 2-paths centered at $v$. The hash table contains a vertex $v$ as a key together with the set of its neighbors $N(v)$. We also store the vertices with only one neighbor in $G_S$ in a hashtable $H_\emptyset$. (Note that if there are no vertices in a given $H_i$, we don't maintain $H_i$.) 

In another hash table $T$ we maintain for each vertex incident to at least one sampled edge, a link to the $H_i$ it is contained in. Whenever a new edge $(u,v)$ is inserted or deleted from $G_S$, we first look-up in $T$ for the $H_i$ containing $u$ and $v$ and then update the corresponding numbers of 2-paths centered at $u$ and $v$. It may happen that we need to move $u$ and/or $v$ from a hashtable $H_i$ to a hashtable $H_{i \pm 1}$. For each $H_i$ we also maintain the total number of 2-paths centered at vertices $v \in H_i$, denote this number as $P_2(H_i)$. Implementing $T$ and the $H_i$ using the implementation from~\cite{opt_hashing}, each update takes constant time with high probability and the total space is $O(m')$.   

We sample a 2-path from $G_S$ at random as follows. We compute $P_2(G_S) = \sum_{i}H_i$ and select an $H_i$ where each $H_i$ has a chance of being picked up of $P_2(H_i)/P_2(G_S)$. This is done by generating a random number $r \in (0,1]$ and then computing a prefix-sum of $P_2(H_i)$'s until the sum reaches $r$ in time $O(\log n)$. Once we have chosen an $H_i$, we select a vertex $v \in H_i$ at random as follows. Assume we maintain the set of vertices in each $H_i$ in a dynamic dictionary $V_i$ implemented as a hashtable using tabulation hashing. We assume that the longest chain in $V_i$ is bounded by a $\kappa = O(\log n/\log \log n)$~\cite{tabulation_hashing}. We select a chain in $V_i$ at random and keep it with probability $\ell/\kappa$, where $\ell \ge 0$ is the length of the selected chain, otherwise we reject it and repeat the step until a chain is kept. Then, we select at random one of the vertices on the chain, let this be $v$. We apply one more time rejection sampling in order to decide whether we keep $v$ or not: Let $q=2^{i+1}-1$. We get the value $d_v$ from $H_i$ and keep $v$ with probability ${d_v \choose 2}/q$ and reject it with with probability $1-{d_v \choose 2}/q$. Once we keep a vertex $v$, we choose at random two of its neighbors in $G_S$ which constitutes the sampled 2-path. The expected number of sampling a chain and a random vertex until we keep a vertex is $O(\log n/\log \log n)$, thus the expected number of trials is $O(\log n)$ and with high probability we determine a 2-path in time $O(\log^2 n)$. It is easy to see that each 2-path is selected with equal probability $p$ such that $\frac{1}{2 \kappa P_2(G_S)} \le p \le \frac{1}{P_2(G_S)}$, thus we sample uniformly at random from the set $P_2(G_S)$. 
\end{proof}
\\\\
Now we have all components in order to prove the main result.
\\
\begin{proof} (of Theorem 1).
 
Assume {\sc EstimateNumberOfTriangles} runs $K$ copies in parallel of {\sc SparsifyGraph} with $p = \frac{5}{\varepsilon \sqrt{b}}$ for $b = \lfloor m/18 \rfloor$. By Lemma~\ref{lem:reducteToBipartite}, Lemma~\ref{lem:bipartite} and and Lemma~\ref{lem:sampl_prob} with probability 3/4 we have a sparsified graph with at least $s = 18/\varepsilon^2$ pairwise independent 2-paths. %Set $t = \frac{20 \sqrt{3} m}{\sqrt{b}}$. Since the number of edges does not exceed $m$, from Markov's inequality it follows that with probability at least $3/4$ there are less than $t$ edges in the sparsified graph. By the union bound it thus follows that with probability at least 1/2 we will sample a 2-path. 
Thus, we expect to obtain from $3K/4$ of them an indicator random variable. A standard application of Chernoff's inequality yields that with probability $O(2^{-K/36})$ we will have $\ell \ge K/2$ indicator random variables $X_i$ denoting whether the sampled 2-path is part of a triangle. By Lemma~\ref{lem:estimate} we have $(1-\varepsilon)\alpha \le \E[X_i] \le (1+\varepsilon)\alpha$ and as an estimate of $\alpha$ we return ${\sum_{i=1}^{\ell} X_i}/{\ell}$. Observe that $(1+ \varepsilon/3)^2 \le 1+ \varepsilon$, respectively $(1-\varepsilon/3)^2 \ge 1 - \varepsilon$. From the above discussion and applying Chernoff's inequality and the union bound, we see that for $K = \frac{36}{\varepsilon^2 \alpha} \log \frac{2}{\delta}$, we obtain an $(\varepsilon, \delta/2)$-approximation of~$\alpha$. 

By Lemma~\ref{lem:second_moment_est} we can compute an $(\varepsilon, \delta/2)$-approximation of the number of 2-paths in space $O(\frac{1}{\varepsilon^2} \log \frac{1}{\delta})$ and $O(\log \frac{1}{\delta})$ per edge processing time. It is trivial to show that this implies an $(3\varepsilon, \delta)$-approximation of the number of triangles for $\varepsilon < 1/3$. Clearly, one can rescale $\varepsilon$ in the above, i.e. $\varepsilon = \varepsilon/3$, such that {\sc EstimateNumberOfTriangles} returns an $(\varepsilon, \delta)$-approximation. 

By Lemma~\ref{lem:sampling}, each sparsified graph with $m'$ edges uses space $O(m')$ and each update takes constant time with high probability, thus we obtain that each edge is processed with high probability in time $O(K)$. Each monochromatic edge and its color can be represented in $O(\log n)$ bits. 

By Lemma~\ref{lem:sampling}, in expected time $O(\log n)$ and worst case time $O(\log^2 n)$ with high probability we sample uniformly at random a 2-path from each $G_S$ with at least $18/\varepsilon^2$ pairwise independent 2-paths. 
\end{proof}

\section{Lower bound}\label{sec:lower}

Pavan et al.~\cite{pavan_et_al} show that every streaming algorithm approximating the number of triangles in adjacency streams (edges are inserted in arbitrary order) needs space $\omega(1/\alpha(G))$.  
For the general case, we show another lower bound on the memory needed which matches the upper bound by Manjunath at al.~\cite{cycles_counting}.
Our lower bound works for the {\em promise version\/} of this problem where the algorithm is required to distinguish between the case where there are no triangles and the case where there are at least $T_3$ triangles, for a parameter $T_3$.
The behavior if the number of triangles is between $1$ and $T_3-1$ is unspecified, i.e., the algorithm may return any result.
Clearly this problem is solved as a special case by any streaming algorithm that is able to approximate the number of triangles, so our space lower bound will also apply to the setting of our upper bound.

\begin{thmx}\label{thm:lb}
Let $G=(V, E)$ be a graph over $m$ vertices.
Any one-pass streaming algorithm distinguishing between the cases where $G$ has at least $T_3$ triangles and $G$ is triangle-free, needs $\Omega(m^3/T_3^2)$ bits.
\end{thmx}
\begin{proof}
We obtain the lower bound by a reduction from 1-way protocols for the {\sc Index} problem in communication complexity.
In this problem Alice is given a bit string $x\in\{0,1\}^a$ and needs to send a message to Bob who holds an index $i$, such that Bob is able to output $x_i$ with probability at least $2/3$ (where the probability is over random coin tosses made by Alice and Bob).
It is known that this problem requires a message of size $\Omega(a)$~\cite{ttt}.

From a streaming algorithm with space usage $s$ that distinguishes the cases of $0$ and $\geq T_3$ triangles we obtain a communication protocol for the indexing problem with $a = \Theta(m^3/T_3^2)$.
The reduction is as follows: Consider $a$ vertex disjoint bicliques $C_1,\dots,C_a$ with $\Theta(T_3/m)$ vertices on each side, and form the stream with edge set $\cup_{i, x_i=1} C_i$.
The number of edges in this stream is $\Theta(a (T_3/m)^2)$, which is $\Theta(m)$.
Adjusting constants we can make the number of edges less than $m/2$.
Alice can simulate the streaming algorithm on this input and send its state of $s$ bits to Bob.
In order to determine the value of $x_i$ Bob now connects $\Theta(m^2/T_3)$ new (i.e., previously isolated) vertices to all vertices in $C_i$.
Adjusting constants this gives another $m/2$ edges, and will create either no triangles (if $C_i$ was not in the stream simulated by Alice) or $\Theta((m^2/T_3)(T_3/m)^2) = \Theta(T_3)$ triangles (if $C_i$ was in the stream).
Thus, if the streaming algorithm succeeds with probability $2/3$, so does $Bob$.
In conclusion we must have a space usage of $s = \Omega(a)$ bits, which is $\Omega(m^3/T_3^2)$ as desired.
\end{proof}

\section{Comparison on the theoretical guarantee for real graphs} \label{sec:dtc_compare}
In Table~\ref{tab:real_values} we compare the theoretical guarantee on the complexity of our algorithm to the ones shown in~\cite{ahn_et_al,cycles_counting} for real graphs. We used the publicly available information from the Stanford Large Network Dataset Collection\footnote{http://snap.stanford.edu/data/index.html}. (Note that there the transitivity coefficient is called {\em fraction of closed triangles.}) As can be seen, the transitivity coefficient $\alpha$ appears to be constant. The update time of $O(n\log n)$ from~\cite{ahn_et_al} is always impractical, and  for~\cite{cycles_counting} it is almost always prohibitively large, the only exception being the Facebook graph for which the  $m^3/T^2_3$ is small. (Note that this is a graph collected from users who used the Social Circles application and does not reflect the structure of the whole social network graph.)
The space usage of our algorithm is also much better for several graphs which are not very dense with respect to the number triangles. 

\begin{table} 
\begin{adjustwidth}{-1.5cm}{}
\begin{center}
\begin{tabular}{ c || p{1.4cm} | p{1.4cm} |  p{1.4cm} | p{1.4cm} | p{1.4cm} | p{1.4cm} | p{1.4cm} }
  \hline                        
Dataset & \hspace{5mm}$n$ & \hspace{5mm}$m$ & \hspace{5mm}$T_3$ & \hspace{5mm}$\alpha$ & \hspace{2mm}${m^3}/{T^2_3}$ & \hspace{2mm}$mn/T_3$ & $\sqrt{m}/ \alpha$\\
  \hline
  \hline
  Enron & \hspace{3mm}36K & \hspace{3mm}367K &  \hspace{2mm}727K & \hspace{2mm}0.0853 & \hspace{2mm}93.5K & \hspace{2mm}{18K} & \hspace{2mm}{\bf 7.1K}\\
  \hline
  AS20000102 & \hspace{3mm}6.4K & \hspace{3mm}13.2K & \hspace{2mm}6.5K & \hspace{2mm}0.0095 &\hspace{2mm}52K & \hspace{2mm}{12.9K} & \hspace{2mm}{\bf 12.1K}\\
  \hline
  Astro-Ph & \hspace{3mm}18.7K & \hspace{3mm}396.1K & \hspace{2mm}1.35M & \hspace{2mm}0.318 & \hspace{2mm}34K & \hspace{2mm}{5.5K} & \hspace{2mm}{\bf 2K}\\
  \hline
  Cond-Mat & \hspace{3mm}23.1K & \hspace{3mm}186.9K & \hspace{2mm}173.3K & \hspace{2mm}0.2643 & \hspace{2mm}215K & \hspace{2mm}24.9K & \hspace{2mm}{\bf 1.6K}\\
  \hline
  roadNet-CA & \hspace{3mm}1.96M & \hspace{3mm}5.53M & \hspace{2mm}120.6K & \hspace{2mm}0.0603 & \hspace{2mm}11.55M & \hspace{2mm}89.8M & \hspace{2mm}{\bf 39K} \\
  \hline
  DBLP & \hspace{3mm}317K & \hspace{3mm}1.05M &  \hspace{2mm}2.2M & \hspace{2mm}0.3064 & \hspace{2mm}239K & \hspace{2mm}148K & \hspace{2mm}{\bf 3.4K}\\
  \hline 
  Oregon & \hspace{3mm}10.6K & \hspace{3mm}22K & \hspace{2mm}17.1K & \hspace{2mm}0.0093 & \hspace{2mm}37.1K & \hspace{2mm}{\bf 13.6K} & \hspace{2mm}16K\\ 
  \hline
  Facebook & \hspace{3mm}4K & \hspace{3mm}88.2K &  \hspace{2mm}1.61M & \hspace{2mm}0.2647 & \hspace{2mm}264 & \hspace{2mm}{\bf 218} & \hspace{2mm}1.1K\\
  \hline
  LiveJournal & \hspace{3mm}4M & \hspace{3mm}34.6M &  \hspace{2mm}177.8M & \hspace{2mm}0.1154 & \hspace{2mm}1.3M & \hspace{2mm}778K & \hspace{2mm}{\bf 51K}\\
  \hline
  Youtube & \hspace{3mm}1.1M & \hspace{3mm}2.9M & \hspace{2mm}3M & \hspace{2mm}0.0062 & \hspace{2mm}2.7M & \hspace{2mm}1.06M & \hspace{2mm}{\bf 275K}\\ 
  \hline 
  Amazon & \hspace{3mm}334K & \hspace{3mm}928K & \hspace{2mm}667K & \hspace{2mm}0.2052 & \hspace{2mm}1.79M & \hspace{2mm}464K & \hspace{2mm}{\bf 4.7K}\\
  \hline
  Skitter & \hspace{3mm}1.7M & \hspace{3mm}11.1M & \hspace{2mm}28.7M & \hspace{2mm}0.0053 & \hspace{2mm}1.66M & \hspace{2mm}657K & \hspace{2mm}{\bf 629K}\\
  \hline
\end{tabular} 
\vspace{1mm}
\caption{Comparison of the theoretical guarantees for several real-life graphs, $b=\max(n, P_2/n)$. The lowest space usage in each row is given in bold font.}
\label{tab:real_values}
\end{center}
\end{adjustwidth}
\end{table}
\section{Conclusions} \label{sec:concl}
We presented a novel algorithm for triangle counting in dynamic graph streams. Since the algorithm is based on sampling it is easy to extend it counting triangle of certain type in directed graphs~\cite{becchetti_et_al}. It is interesting whether similar techniques can be applied to counting more complex graph minors in dynamic graph streams like in~\cite{subgraph_counting}.

Another open question is to obtain lower bounds on the streaming complexity of triangle counting that reflect better the requirements for real graphs. The lower bound shown in Section~\ref{sec:lower} matches the upper bound from~\cite{subgraph_counting} for certain graph classes. However, as evident from the values in Table~\ref{tab:real_values}, the triangle distribution for real graphs allows more space-efficient algorithms.  

%\newpage

%
%
%
\end{document}